\documentclass[letterpaper, 11pt]{article}
\usepackage{graphicx,subfigure} 
\usepackage{epsfig,color} 
\usepackage{amsmath,algorithm,algorithmic} 
\usepackage{amssymb,mathrsfs}  
\usepackage{theorem}
\usepackage{fancyhdr}
\usepackage[breaklinks]{hyperref}
\usepackage[margin=1.0in]{geometry}
\hypersetup{pdfborder={0 0 0},
            breaklinks=true,%
            bookmarks=true,%
            bookmarksnumbered=true,%
            colorlinks=true,%
            citecolor=red,%
            linkcolor=blue,%
}

\fancypagestyle{plain}{%
\fancyhf{}
\cfoot{}

}

\newenvironment{proof}{\noindent{\bf Proof:}\rm}{\hfill\hfill$\Box$\par\medbreak}

 \newtheorem{theorem}{Theorem}[section]
 \newtheorem{proposition}[theorem]{Proposition}

\newtheorem{remark}[theorem]{Remark}

\newcommand{\real}{\mathbb{R}}
\newcommand{\e}{\mathbf{e}}

\newcommand{\X}{\mathbf{X}}

\newcommand{\E}{\mathbb{E}}

\newcommand{\R}{\mathbf{R}}

\newcommand{\Tr}{\operatorname{Tr}}

\newcommand{\diag}{\operatorname{diag}}
\newcommand{\blkdiag}{\operatorname{blkdiag}}
\newcommand{\Prob}{\operatorname{Prob}}

\newcommand{\subscr}[2]{{#1}_{\textup{#2}}}
\newcommand{\norm}[1]{\|#1\|}

\long\def\cut#1{{}}
\newcommand\sdb[1]{{\color{black}{#1}}} 

\begin{document}

\title{Sensor Selection via Randomized Sampling \thanks{The author is with the Department of Electrical and Computer Engineering, Michigan State University, East Lansing, MI, USA. Email: \texttt{shaunak@egr.msu.edu}}}

\author{Shaunak D. Bopardikar}

\maketitle
\thispagestyle{plain}

\begin{abstract}
Given a linear dynamical system, we consider the problem of constructing an approximate system using only a subset of the sensors out of the total set such that the observability Gramian of the new system is approximately equal to that of the original system. Our contributions are as follows. First, we present a randomized algorithm that samples the sensors with replacement as per specified distributions. For specific metrics of the observability Gramian such as the trace or the maximum eigenvalue, we derive novel bounds on the number of samples required to yield a high probability lower bound on the metric evaluated on the approximate Gramian. Second, with a different distribution, we derive high probability bounds on other standard metrics used in sensor selection, including the minimum eigenvalue or the trace of the Gramian inverse. This distribution requires a number of samples which is larger than the one required for the trace and the maximum eigenvalue, but guarantees non-singularity of the approximate Gramian if the original system is observable with high probability. Third, we demonstrate how the randomized procedure can be used for recursive state estimation using fewer sensors than the original system and provide a high probability upper bound on the initial error covariance. We supplement the theoretical results with several insightful numerical studies and comparisons with competing greedy approaches.

\end{abstract}

\section{Introduction}
Sensor selection is a classic problem arising in engineering design wherein the goal is to choose sensors out of a set in order to estimate the quantity of interest. In applications related to complex cyber physical systems, it is essential to understand fundamental limits on how the system performance varies with the number of samples of sensors. This paper considers the problem of analyzing observability of a linear time invariant system as a function of the number of samples of sensors. We present a randomized algorithm and derive novel probabilistic bounds on the observability of the resulting system.

\subsection{Related work}
\sdb{Sensor selection is a well studied problem within the control community (see survey \cite{van2001}). }Early work on sensor selection includes reference \cite{bian:06}, which established NP-hardness of a sensor selection problem and reference \cite{joshi:09}, which provided a heuristic based on a convex optimization formulation of another sensor selection problem. However, no guarantees were provided on the performance of the heuristics. There has been a lot of recent interest in the area of control of complex networks, especially on the dual problem of finding a subset of control inputs to steer the state of a dynamical system to a target state. Reference \cite{pasqualetti:14} proposes several metrics and a formal mathematical framework and examines tradeoffs between number of nodes controlled and the control energy required. Reference \cite{olshevsky:14} demonstrates that the problem of finding a small set of variables to affect with an input which guarantees controllability of the resulting system is NP-hard and provides a greedy heuristic that has good empirical performance. Reference \cite{summers:16} identifies several metrics, including the ones considered in this paper and demonstrates sub/super-modularity property, and further shows that the use of a simple greedy method can solve the combinatorial problem of the minimum number of sensors required to maximize the selected criterion within a constant $(1-1/e)$ factor of the optimal set. Reference \cite{pequito:16} addresses the problem of optimal input/output placement  that guarantees structural controllability/observability by reducing the problem to a weighted maximum matching problem and solves that using existing efficient polynomial time algorithms. Reference \cite{tzoumas:16} addresses the problem of optimal sensor placement to design a Kalman filter with a desired bound on the estimation error covariance, derives fundamental limitations, proves super-modularity of the log-determinant function and provides an efficient approximation algorithm. In the context of Kalman filtering, reference \cite{jawaid:15} addresses optimal sensor scheduling to select a subset of sensors at every measurement step to minimize the estimation error at the subsequent time step. Reference \cite{zhang:17} examines the trace of the a priori and the a posteriori error covariance matrices, shows that the optimization problem is NP-hard and characterizes systems for which greedy algorithms are optimal. With a fixed budget on the number of sensors that can be used, recent work~\cite{hashemi2018} addresses the use of randomization for sensor scheduling and uses weak submodularity to performance guarantees on the mean square error (MSE) of the linear estimator. 

\sdb{Several metrics for quantifying the degree of observability have been studied in literature. Early reference \cite{muller1972} proposes the minimum eigenvalue of the observability Gramian that quantifies modes that are difficult to estimate from the measurements. However, this worst case metric may lead to impractical solutions in applications such as tubular reactors as reported in reference \cite{van2000}. If the criterion is to monitor only major changes in the state variables, then reference \cite{van2000} proposes the use of the trace or the maximum eigenvalue of the observability Gramian as metric. This paper also considers the problem of selecting a subset of the sensors so that the resulting system has approximately similar observability Gramian as the original, when measured in the trace or the maximum eigenvalue metric. } 

\subsection{Contributions}
This paper presents a randomization-based paradigm for the problem of sensor selection. Given a linear dynamical system, we consider the problem of \sdb{randomly sampling} sensors out of the total set. Our contributions are as follows. 

First, we formalize the performance of a simple randomized sampling algorithm and derive a high probability lower bound on two specific metrics of the observability Gramian, the trace and the maximum eigenvalue. We provide a bound on the number of samples required by the algorithm and we characterize the expected number of unique sensors being selected.  

However, as is well known, simply providing a guarantee on these two metrics does not guarantee non-singularity of the approximate observability Gramian, unless the system is observable through every sensor. Therefore, as a second contribution, we design a new distribution to sample the sensors so that any other standard metric of the eigenvalues of the observability Gramian, such as the minimum eigenvalue or the trace of the inverse of the Gramian can be well approximated. The accuracy is relative to the \emph{global value} obtained by using \emph{all} of the available sensors. This is a key distinction with respect to the sensor selection literature where the performance has largely been characterized relative to the optimal (but difficult to compute) subset of a given cardinality.  Our procedure guarantees that the resulting approximate Gramian is non-singular with high probability, if the original system is observable. 

Third, we demonstrate how the proposed randomized procedure can be used for recursive least squares estimation using fewer sensors than the original system and provide a high probability upper bound on the initial error covariance relative to that for the original system. We supplement our theoretical bounds with numerical results on a synthetic dataset. In particular, we focus on: (i) quantitative validation of the theoretical results on a numerical example using the minimum eigenvalue, the trace and the maximum eigenvalue as metrics, (ii) comparison with a greedy heuristic that selects the sensors based on the contribution of each sensor toward the chosen metric and (iii) empirical study of application of the technique for Kalman filtering. The empirical performance of the proposed approach shows significant improvement over the greedy heuristic. 

 The randomization viewpoint has been used extensively to solve complex control design problems over the past two decades, e.g., references \cite{TempoBaiDabbene97}, \cite{vidyasagar1998} and \cite{calafiore2006} to name a few. These techniques adopt a sampling based approach and provide probabilistic bounds on the resulting controller performance. In contrast, our approach leverages insights and analysis techniques from the field of random matrix theory. We refer the interested reader to the tutorial paper \cite{tropp:11b} for a detailed survey in the field. In this context, reference \cite{preciado:16} studies the eigenvalue spectra of the controllability Gramian of systems with random state matrices. The computation of the Gramian can be approximated efficiently using randomized sampling as was demonstrated in reference \cite{drineas:05}. The key contribution of this paper is to develop a novel randomized sampling scheme that can be used to select the sensors and provide guarantees on the observability Gramian of the resulting system with respect to the Gramian of the original system with high probability. A key feature of our approach is that it is \emph{one-shot} in the sense that a user samples just once a priori and the sampling outcome provides an expression for the approximate Gramian leading to probabilistic guarantees. This paradigm offers a novel alternative to the sensor selection formulations and techniques in literature and addresses difficult to analyze metrics such as the minimum eigenvalue of the observability Gramian.

\subsection{Organization of this paper}
This paper is organized as follows. Section~\ref{sec:problem} provides a formal statement of the problem and provides a brief background on the metrics used. Section~\ref{sec:algorithm} formalizes our algorithm and includes the main theoretical results. Section~\ref{sec:estimation} discusses application of the proposed approach to recursive least squares estimation and other related applications. Section~\ref{sec:numerics} summarizes numerical results of the proposed approach on synthetically generated data. Section~\ref{sec:conclusion} summarizes our findings and identifies directions for future research. The mathematical proofs of the theoretical results are presented in the appendix.

\section{Problem Formulation and Background}\label{sec:problem}
In this section, we formalize the problem statement and provide some background information related to metrics on the observability Gramian.

\subsection{Model}
Consider a linear time invariant system described by
\begin{align}\label{eq:model}
x_{t+1} &= Ax_t + w_t, \nonumber \\
y_t &= Cx_t + v_t,
\end{align}
where at any given time instant $t$, $x_t \in \real^n$ is the state, $n$ is the dimension of the state space, $y_t \in \real^m$ are the measurement from all of the sensors, $w_t \sim \mathcal{N}(0, Q)$ is the process noise, $v_t \sim \mathcal{N}(0, R)$ is the measurement noise and $m$ is the number of sensors. We assume that $Q$ is positive semi-definite and $R$ is positive definite.

A sample of sensors is a subset $S \subseteq \{1,\dots, m\}^c$, with possibly repeated indices allowing for sampling with replacement. Here, $c$ is the number of samples. Each realization of $S$ can be mapped to a sensor selection matrix $\Gamma \in \mathcal{B}^{m\times m}$ in which the $i$-th diagonal entry is equal to the number of times sensor $i$ appears in $S$ and the off-diagonal entries of $\Gamma$ equal zero. To analyze the observability of system~\eqref{eq:model} with a sensor selection matrix $\Gamma$, we consider the $T$-step observability Gramian,
\begin{equation}\label{eq:gramian}
W_{\Gamma, T} := \sum_{t = 0}^{T-1} (A')^tC' \Gamma C A^t.
\end{equation}

\subsection{Metrics on the Observability Gramian}
It is known that $W_{\Gamma, T}$ quantifies the ability to observe the system state from the output, as shown in reference~\cite{van2001}. In particular, the \emph{maximal} energy of the output obtained by observing a system with initial state $x_0$ is $x_0'W_{\Gamma, T} x_0$ and is measured by the maximum eigenvalue $\subscr{\lambda}{max}(W_{\Gamma, T})$. The trace, $\Tr(W_{\Gamma, T})$ represents the \emph{average} output energy over random initial states. Constructing the dual of system~\eqref{eq:model}, the inverse of the observability Gramian of \eqref{eq:model} represents the minimal energy required to transfer the state of the dual system from $0$ to a given target state. In this case, akin to reference~\cite{summers:16}, the problem becomes one of selecting a $\Gamma$ such that $W_{\Gamma, T}$ is \emph{large} so that $W_{\Gamma, T}^{-1}$ is \emph{small} in some appropriate metric. Early work (see reference~\cite{muller1972}) has shown that $\Tr(W_{\Gamma, T}^{-1})$ represents the \emph{average} energy over random target states. The following relation,
\[
\frac{\Tr(W_{\Gamma, T}^{-1})}{n} = \sum_{i=1}^n \lambda_i^{-1}(W_{\Gamma, T}) = \sum_{i=1}^n \frac{1}{\lambda_i(W_{\Gamma, T}) } \leq \frac{n}{\subscr{\lambda}{min}(W_{\Gamma, T})},
\]
motivates the use of $\subscr{\lambda}{min}(\cdot)$ as an alternative metric considered in this paper. In other words, maximizing $\subscr{\lambda}{min}(\cdot)$ implies minimizing an upper bound on $\Tr(W_{\Gamma, T}^{-1})$. In this paper, we will focus on analyzing the three metrics, $\subscr{\lambda}{max}(\cdot), \Tr(\cdot)$ and $\subscr{\lambda}{min}(\cdot)$. 

\subsection*{Problem Statement}
This paper considers a formulation in which the sensor set $S$ is chosen so that the resulting subsystem $(A, \Gamma(S)C)$ and the original system $(A,C)$ have approximately similar observability Gramians. The primary goal of this paper is to design an algorithm to select $S$, along with a theoretical characterization of the number of samples $c$ required to solve the feasibility problem,
\begin{align*}
&\min_{S \in \{1,\dots, m\}^c} 0 \\
&\text{subject to } \mathcal{M}(W_{\Gamma(S)}) \geq (1-\epsilon) \mathcal{M}(W_{I}),
\end{align*}
where $\epsilon \in (0,1)$ is given and $\mathcal{M}$ denotes one of the metrics listed in the previous subsection, viz.~$\subscr{\lambda}{max}(\cdot)$, $\Tr(\cdot)$ or $\subscr{\lambda}{min}(\cdot)$, and $I$ is the $m\times m$ identity matrix.

While a greedy algorithm is simple solution for the trace metric due to linearity of trace, the problem is difficult for the other metrics. This paper proposes a randomized approach to select the set $S$. Thus, the sensor selection matrix $\Gamma$ and therefore, $W_{\Gamma}$ are random variables. So we will solve the following modified probabilistic feasibility problem,
\begin{align}\label{eq:metric}
&\min_{S \in \{1,\dots, m\}^c} 0 \nonumber \\
&\text{subject to } \Prob\Big(\mathcal{M}(W_{\Gamma(S)}) \geq (1-\epsilon) \mathcal{M}(W_{I}) \Big) \geq 1-\delta,
\end{align}
where $\delta >0$ is a specified parameter. 

In general, we are interested in the approximating the Gramian so that it is sufficiently close to the original value in the positive definite sense. Formally, this can be expressed as the following feasibility problem.
\begin{align*}
&\min_{S \in \{1,\dots, m\}^c} 0 \\
&\text{subject to }  (1-\epsilon) W_I \preceq W_{\Gamma(S)} \preceq (1+\epsilon) W_{I},
\end{align*}
for a given $\epsilon \in (0,1)$ where $I \in \real^{m\times m}$ is the identity and given any two positive semidefinite matrices $A, B$, $A \preceq B$ denotes that the matrix $B-A$ is positive semidefinite.

Given the probabilistic setting considered in this paper, a second goal of this paper is to solve the modified probabilistic feasibility problem given by
\begin{align}\label{prob}
&\min_{S \in \{1,\dots, m\}^c} 0 \nonumber \\
&\text{subject to } \Prob \Big((1-\epsilon) W_I \preceq W_{\Gamma(S)} \preceq (1+\epsilon) W_{I} \Big) \geq 1-\delta.
\end{align}
Note that a solution to Problem~\ref{prob} also yields a solution to Problem~\ref{eq:metric} for any specific metric out of the set considered. 

A third goal of this paper is to apply the solution to Problem~\ref{prob} to least squares estimation of the state of the system using a least squares observer and characterize the performance in terms of the initial error covariance. In what follows, we will present the solutions to these three problems sequentially.

%

\section{Randomized Algorithm and Main Results} \label{sec:algorithm}
In this section, we will propose a randomized algorithm for sensor selection and then analyze its performance from an accuracy and computational complexity perspective. \sdb{Since $T$ is fixed, we will drop the explicit $T$ dependence in the notation used for the Gramian $W$.

\subsection{The Algorithm}
We begin with \eqref{eq:gramian} for the case with $\Gamma = I_{m\times m}$. We drop this subscript for the sake of convenience.}
\begin{align*}
W &=  \sum_{t = 0}^{T-1} (A')^tC'C A^t \\ &= \sum_{t=0}^{T-1} (A')^t (\sum_{k=1}^m c_k'c_k ) A^t \\ &= \sum_{k=1}^m \sum_{t=0}^{T-1} (A')^t c_k'c_k A^t =: \sum_{k=1}^m W_k,
\end{align*}
where for each $k \in \{1,\dots, m\}$, 
\begin{equation}\label{eq:xi}
W_k = \sum_{t=0}^{T-1} (A')^t c_k'c_k A^t  \in \real^{n \times n},
\end{equation}
and $c_k$ denotes the $k$-th row of $C$. Then, Algorithm~\ref{algo:randsensor} summarizes our randomized approach to sample the sensors. \sdb{Note that step 6 involves sampling with replacement and therefore, the number of unique sensors chosen may be less than the number sampled $c$ and is certainly less than the total number of sensors $m$.}

\begin{algorithm}[h]
\begin{algorithmic}[1]
\STATE \textbf{Input:} Matrices $W_i , \forall i \in \{1,\dots, m\}$ defined in \eqref{eq:xi}, number of samples of sensors $c$.
\FOR{$i = 1, 2, \dots, m$}{
\STATE Choose a value for $p_i$ (to be set in the statements of Theorems~\ref{thm:randsensor2}~\ref{thm:randsensor3} and~\ref{thm:randsensor}).}
\ENDFOR
\STATE Set, $p_i = p_i/\sum_{i=1}^m p_i, \forall i \in \{1,\dots, m\}$.
\STATE Sample $c$ columns $j_1, \dots, j_c$ out of $m$ as per probability $p_1, \dots, p_m$ and \emph{with replacement}.
\STATE \textbf{Output: } Approximate Gramian, \\
\qquad \qquad $\displaystyle G := \frac{1}{c}\sum_{i=1}^c \frac{1}{p_{j_i}}W_{j_i}$
\end{algorithmic}
\caption{Randomized Sensor Selection} \label{algo:randsensor}
\end{algorithm}  

Since Algorithm~\ref{algo:randsensor} involves sampling with replacement, it is anticipated that one or more sensors get chosen multiple number of times. Therefore, although the number of samples $c$ may be large, the number of \emph{unique} sensors selected will typically be smaller, compared to the number of sensors $m$. Since the number of unique sensors is a random variable, the following result characterizes its expected value.

\begin{proposition}[Number of unique sensors]\label{thm:unique}
The expected number of unique sensors selected as a result of Algorithm~\ref{algo:randsensor} is 
\[
m - \sum_{i=1}^m\big( 1- p_i\big)^c.
\]
In the special case when $p_i = 1/m, \forall i \in \{1,\dots, m\}$, the expected number of unique sensors tends to $m(1-\e^{-m/c})$ in the limit of large $m$.
\end{proposition}

\subsection{Main results: Sample Complexity Bounds}
Note that the approximate Gramian $G$ is a random variable. Theorem~\ref{thm:randsensor2} characterizes how accurately $G$ approximates the Gramian of the entire system $W$ using $\subscr{\lambda}{max}(\cdot)$ as metric, in a probabilistic sense. 

\begin{theorem}[Maximum eigenvalue as metric]\label{thm:randsensor2}
Let $\epsilon \in (0,1)$ and $\delta \in (0,1)$ and let $G$ denote the output of Algorithm~\ref{algo:randsensor} for the system~\eqref{eq:model}. With the choice of $p_i:= \subscr{\lambda}{max}(W_i)$ and 
\[
c \geq 2.7\frac{\sum_{k=1}^m\subscr{\lambda}{max}(W_k)}{\epsilon^2\subscr{\lambda}{max}(W)}\log\Big (\frac{n}{\delta}\Big ),
\]
we have 
\[
\Prob\big(\subscr{\lambda}{max}(G) \geq (1-\epsilon)\subscr{\lambda}{max}(W) \big) \geq 1-\delta.
\]
\end{theorem}

\sdb{In this result, $\epsilon$ is the accuracy parameter while $\delta$ controls the probability with which the bound holds. Clearly, if both $\epsilon, \delta$ are close to zero, then the number of samples $c$ required by this result increases.} The number of samples $c$ depends greatly on the numerical value of  $\sum_{k=1}^m\subscr{\lambda}{max}(W_k)/\subscr{\lambda}{max}(W)$. \sdb{Numerical studies in Section~\ref{sec:numerics} will shed further light on how practical this bound is.}

The requirement on the sample complexity can be made extremely sharp for the choice of the trace of the Gramian as the metric in \eqref{eq:metric}. The following result provides a guarantee on the trace of the Gramian.

\begin{theorem}[Trace of the approximate Gramian] \label{thm:randsensor3}
Let $G$ denote the output of Algorithm~\ref{algo:randsensor} for the system~\eqref{eq:model} and let $\epsilon \in (0,1)$ be given. With the choice of $p_i := \Tr(W_i)$ and for \emph{any} $c \in \{1,\dots, m\}$, we have
\[
\Prob\big(\Tr(G) \geq (1-\epsilon)\Tr(W)\big) = 1.
\]
\end{theorem}

This result is intuitive in the sense that due to the linearity property of the trace operator, it is straightforward to see that $\E[\Tr(W-G)] = \Tr(\E[W-G]) = 0,$ with the particular choice of $p_i$. It takes additional work to show that the variance of $\Tr(W-G)$ also equals zero. 

Theorems~\ref{thm:randsensor2} and \ref{thm:randsensor3} apply only to the respective specific metrics and do not provide any guarantees on whether the resulting approximate Gramian $G$ is non-singular with high probability, unless there is additional structure in the system, such as observability through each individual sensor. For guaranteed non-singularity of $G$, we need to construct a different sampling distribution described as follows.

For each sensor $k \in \{1,\dots, m\}$, we can associate a positive scalar $\gamma_k \in [0,1]$, which quantifies the amount by which the Gramian of the total system $W$ needs to scaled so that the following inequality holds.
\begin{equation}\label{eq:gamma}
\gamma_k := \min\{ \gamma \in \real \, : \, W_k \leq \gamma W \}.
\end{equation}

If $W$ is invertible, then $\gamma_k>0$ and is given by
\[
\gamma_k = \subscr{\lambda}{max}(W^{-1} W_k).
\]

Then, the following guarantees hold on the output of Algorithm~\ref{algo:randsensor}.

\sdb{\begin{theorem}[Approximate Gramian accuracy]\label{thm:randsensor}
Let $\epsilon \in (0,1)$, $\delta \in (0,1]$ and let $G$ denote the output of Algorithm~\ref{algo:randsensor} for the system~\eqref{eq:model}. Then, with the choice of $p_i = \subscr{\lambda}{max}(W^{-1} W_i)$ and 
\[
c \geq \frac{4\sum_{k=1}^m \gamma_k}{\epsilon^2}\log\frac{2n}{\delta},
\]
we have 
\[
\Prob\Big( (1-\epsilon)W \preceq G \preceq  (1+\epsilon)W \Big ) \geq 1-\delta.
\]
\end{theorem}}


\begin{remark} \label{rem:lmax} It can be easily shown that the sample complexity bound in Theorem~\ref{thm:randsensor2} is lower than that in Theorem~\ref{thm:randsensor} as follows. Given the constants, it only remains to be checked that
\[
\sum_{k=1}^m \subscr{\lambda}{max}(W^{-1} W_k) \geq \frac{\sum_{k=1}^m\subscr{\lambda}{max}(W_k)}{\subscr{\lambda}{max}(W)}.
\]
Using the sub-multiplicativity of $\subscr{\lambda}{max}(\cdot)$, we conclude that 
\begin{align*}
\sum_{k=1}^m \subscr{\lambda}{max}(W^{-1} W_k) \subscr{\lambda}{max}(W) &\geq \sum_{k=1}^m \subscr{\lambda}{max}(W^{-1} W W_k) \\
&= \sum_{k=1}^m \subscr{\lambda}{max}(W_k).
\end{align*}
Although the sample complexity requirement of Theorem~\ref{thm:randsensor} is higher, its strength lies in the fact that it can be used to perform sensor selection using \emph{any} of the eigenvalue based worst case metrics such as $\subscr{\lambda}{min}(\cdot)$ or the trace of the inverse of the Gramian.
\end{remark}

In terms of computational complexity, Algorithm~\ref{algo:randsensor} requires one pass through all of the $m$ sensors to compute the sampling probabilities. Thereafter, the computation of $G$ is $O(cTn^2)$. The bounds in Theorems~\ref{thm:randsensor} and \ref{thm:randsensor2} depend on the properties of the sensors but the \emph{number of sensors $m$ does not explicitly appear.} It now remains to be seen how Algorithm~\ref{algo:randsensor} can be used to construct the corresponding approximate linear system, how practical the sample complexity bounds turn out to be, and how Algorithm~\ref{algo:randsensor} fares against simple greedy algorithms. This will be the  focus of the next two sections.

\section{Applications: Least-squares Estimation}\label{sec:estimation}
We will now discuss the application of the proposed approach to the problem of recursive least-squares estimation. For the system~\eqref{eq:model} with $Q = 0$ (i.e., no process noise), it is well-known (see reference~\cite{boyd}) that the least squares estimate of the initial state is given by
\[
\subscr{\hat{x}}{LS} = W^{-1} O' \begin{bmatrix} y_0 \\ y_1 \\ \vdots \\ y_T \end{bmatrix},
\]
where 
\[
O := \begin{bmatrix}  C \\  CA \\ \vdots \\ CA^{T-1} \end{bmatrix} \in \real^{mT\times n},
\]
is the \emph{observability matrix}. The estimation error $e := \subscr{\hat{x}}{LS} - x_0$ is given by the random variable $\mathcal{N}(0, \Sigma)$, where
\[
\Sigma = W^{-1} O' \R OW^{-1},
\] 
and $\R \in \real^{mT \times mT}$ is the block diagonal matrix with the matrix $R$ along the block diagonals, i.e.,
\[
\R := \blkdiag(R, R, \dots, R).
\]

Algorithm~\ref{algo:randsensor} and Theorem~\ref{thm:randsensor} can be used to construct an approximate system $A, \bar{C}$, such that the observability Gramians of $(A, C)$ and $(A,\bar{C})$ solve Problem~\ref{prob}.  The measurement vector $y_t$ and the noise covariance $R$ also need to be modified to $\bar{y}_t$ and $\bar{R}$ to be compatible with the approximate system $A, \bar{C}$. 

Specifically, suppose that Algorithm~\ref{algo:randsensor} returns a sensor selection in which row $i$ has been selected $n_i$ times, for every $i\in \{1,\dots, m\}$. Let the number of unique sensors be $q$. Then,
\begin{enumerate}
\item Set the $i$-th row $\bar{c}_i$ of $\bar{C}$ using the following rule:
\[
\bar{c}_{i} = \sqrt{\frac{n_i}{cp_i}}c_{i}, \forall i \in \{1,\dots, m\}.
\] 
\item Set the $i$-th entry $\bar{y}_t(i)$ of $\bar{y}$ using the rule:
\[
\bar{y}_{t}(i) = \sqrt{\frac{n_i}{cp_i}}y_{t}(i), \forall i \in \{1,\dots, m\}.
\] 
\item Set the $i,j$-th entry $\bar{R}_{ij}$ of $\bar{R}$ using the rule:
\[
\bar{R}_{ij} = \sqrt{\frac{n_in_j}{cp_ip_j}}R_{ij},  \forall (i,j) \in \{1,\dots, m\}\times\{1,\dots, m\}.
\] 
\item Remove the rows that contain all zeros out of $\bar{C}$ and $\bar{y}$ and the rows and columns that contain all zeros out of $\bar{R}$.
\end{enumerate}
This procedure is illustrated in Figure~\ref{fig:application}. 

\begin{figure}[!h]
\centering
\includegraphics[width=0.6\columnwidth]{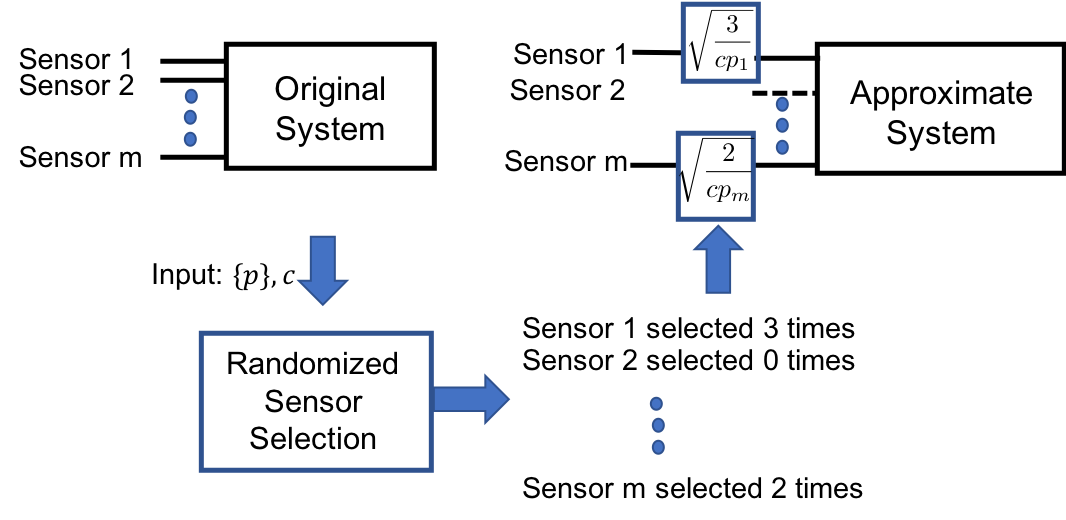}
\caption{Illustration of how Algorithm~\ref{algo:randsensor} can be used to synthesize the approximate system. The output of each sensor will simply be scaled by the factor in line 7 of  Algorithm~\ref{algo:randsensor}. If a sensor does not get chosen (e.g., sensor 2 in this figure), then the output of that sensor is not used in the modified system.}
\label{fig:application}
\end{figure}

These relations can be written compactly as $\bar{y} = \Pi y, \bar{C} = \Pi C, \bar{R} = \Pi R \Pi'$, where the matrix $\Pi\in \real^{q\times m}$ is full row rank and exactly one entry in every row is non-zero and is given by the term $\sqrt{\frac{n_j}{cp_j}}$, corresponding to the $j$-th sensor getting selected. Thus, if $i_1, i_2, \dots, i_q$ are the indices corresponding to $q$ distinct sensors that get selected, we have
\[
\Pi \,\Pi' = \diag \Big (\frac{n_{i_1}}{cp_{i_1}},\dots, \frac{n_{i_q}}{cp_{i_q}} \Big ) \in \real^{q\times q}.
\]

It is straightforward to verify that the observability Gramian of $A,\bar{C}$ is the $T$-step observability Gramian $G$ from  Algorithm~\ref{algo:randsensor}, and the observability matrix $\bar{O} = {\bf \Pi} O$, where
\[
{\bf \Pi} := \blkdiag(\Pi, \Pi, \dots, \Pi)  \in \real^{qT \times mT}.
\]
In particular, we have
\[
\bar{y}_t = \bar{C}x_t + \bar{v}_t,
\]
where $\bar{v}_t \sim \mathcal{N}(0, \bar{R})$. Analogous to $\R$, we can then define the block diagonal matrix, $\bar{\R} := {\bf \Pi} \R {\bf \Pi}'$. From Theorem~\ref{thm:randsensor}, we conclude that the system $A,\bar{C}$ is observable with probability at least $1-\delta$ and the minimum error variance $\bar{\sigma}$ satisfies
\begin{align*}
\bar{\Sigma} &= G^{-1} \bar{O}'\bar{\R}\bar{O}G^{-1} \\
&\preceq \frac{1}{(1-\epsilon)^2} W^{-1} O' \,{\bf \Pi}' {\bf \Pi} \, \R \, {\bf \Pi}' {\bf \Pi} \, O W^{-1} \\
&\preceq \frac{\subscr{\lambda}{max}(\Pi' \Pi R \Pi' \Pi)}{(1-\epsilon)^2} W^{-1} O' OW^{-1}\\
 &= \frac{\subscr{\lambda}{max}(\Pi' \Pi R \Pi' \Pi)}{(1-\epsilon)^2} W^{-1}.
\end{align*}

In particular, if $R$ is diagonal, then the above bound simplifies to 
\[
\bar{\Sigma} \leq \frac{\max_{i}\{\frac{n_iR_{ii}}{cp_i}\}}{(1-\epsilon)^2} W^{-1}.
\]

The following result summarizes the steps in this sub-section into a formal guarantee.

\begin{proposition}\label{prop:estimation}
For the approximate system $(A, \bar{C})$ constructed using Algorithm~\ref{algo:randsensor} with the modified measurements $\bar{y}$ and the covariance $\bar{R}$, the recursive least squares estimation error covariance satisfies
\[
\bar{\Sigma} \preceq \frac{\subscr{\lambda}{max}(\Pi' \Pi R \Pi' \Pi)}{(1-\epsilon)^2} W^{-1},
\]
with probability at least $1-\delta$.
\end{proposition}

\begin{remark}[Randomized Actuator Selection]
Given the duality between observability and controllability \cite{summers:16}, one can pose a problem similar to that considered in this paper for actuator selection. Algorithm~\ref{algo:randsensor} easily extends to select actuators out of a given set in order to optimize the same metrics for the controllability Gramian. \sdb{In this case, multiple samples of an actuator will be addressed in a manner analogous to Figure~\ref{fig:application} and the procedure detailed in this section by appropriately scaling the corresponding row of the control matrix $B$.} Claims analogous to Theorems~\ref{thm:randsensor}, \ref{thm:randsensor2} and \ref{thm:randsensor3} and Proposition~\ref{thm:unique} will hold on the number of samples to be drawn.
\end{remark}

\begin{remark}[Kalman Filtering]\label{rem:kf}
The proposed procedure can also be used to design a Kalman filter to estimate the state when $Q \succeq 0$ and $(A, \sqrt{Q})$ is controllable. Theorem~\ref{thm:randsensor} guarantees that the system $(A, \bar{C})$ is observable with probability $1-\delta$. Therefore, the Kalman filter converges to yield the steady state error covariance matrix $\bar{P}$, which satisfies the following equation.
\[
\bar{P} = A\bar{P}A' + Q - A\bar{P}C'\Pi' ( \Pi (C \bar{P} C' + R) \Pi' )^{-1} \Pi C \bar{P} A'.
\]
The steady state error covariance matrix $P$ satisfies the corresponding equation for the entire system, i.e.,
\[
P = APA' + Q - A{P}C' ( C {P} C' + R) )^{-1} C {P} A'.
\]
It now remains to be seen how close the steady state error covariances $\bar{P}$ and $P$ get. This will be the subject of one of the numerical studies in Section~\ref{sec:numerics}.
\end{remark}

\section{Numerical results} \label{sec:numerics}
In this section, we report numerical results of implementing Algorithm~\ref{algo:randsensor} on a class of linear dynamical systems and empirically study the sample complexity bounds from Section~\ref{sec:algorithm}. We also present a numerical comparison between two variants of Algorithm~\ref{algo:randsensor} and a simple greedy heuristic.

\subsection{System model} 
In our experiments, we consider an $n=100$ dimensional state space and $m = 100$ sensors. We assume that the system is in the \emph{observability canonical form} in which the entries in the last column of $A$ are selected independently and uniformly randomly from the interval $[-1,0]$, while the entries in $C$ are chosen independently and uniformly randomly from the interval $[0,1]$. Such a choice guarantees that the pair $(A,c_i), \forall i\in \{1,\dots, m\}$ is observable. This realization was kept fixed throughout the experiments. 

\subsection{Validation of Algorithm~\ref{algo:randsensor}}
We first validated our approach for a fixed value of $\delta = 0.1$ and a range of values of $\epsilon$. We estimated the observability Gramian using Algorithm~\ref{algo:randsensor} after running 100 Monte Carlo trials. Since our algorithm relies on sampling \emph{with replacement}, there are instances when the number of unique sensors selected by our randomized algorithm is much less than the number of samples $c$. In our experiments, we studied two variants -- the first one in which repetitions were allowed leading to possibly multiple instances of a sensor, and the second in which only single instance of the unique sensors was allowed. We evaluated the performance on the three metrics: the minimum eigenvalue, the trace and the maximum eigenvalue of the observability Gramian. 

Figure~\ref{fig:mineig} validates Theorem~\ref{thm:randsensor} (the solid blue line is always above the dashed red curve) and further suggests that while the performance of the second variant may be lower than the bound, beyond a non-trivial value of $\epsilon (\epsilon > 0.6)$, the second variant also satisfies requirement \eqref{eq:metric}. Figure~\ref{fig:unique} compares the number of unique sensors in the first variant with the bound from Proposition~\ref{thm:unique} and we see an overlap of the two curves as is expected and with a very small variance. Figure~\ref{fig:samples} compares the number of samples of sensors using the two variants of Algorithm~\ref{algo:randsensor}. For the first variant, this is the bound on $c$ given by Theorem~\ref{thm:randsensor} and for the second, it is given by Proposition~\ref{thm:unique}. The results suggest that for $\epsilon > 0.8$, the number of samples $c < m$. 

\begin{figure}[h]
    \centering
           \includegraphics[width = 0.6\columnwidth]{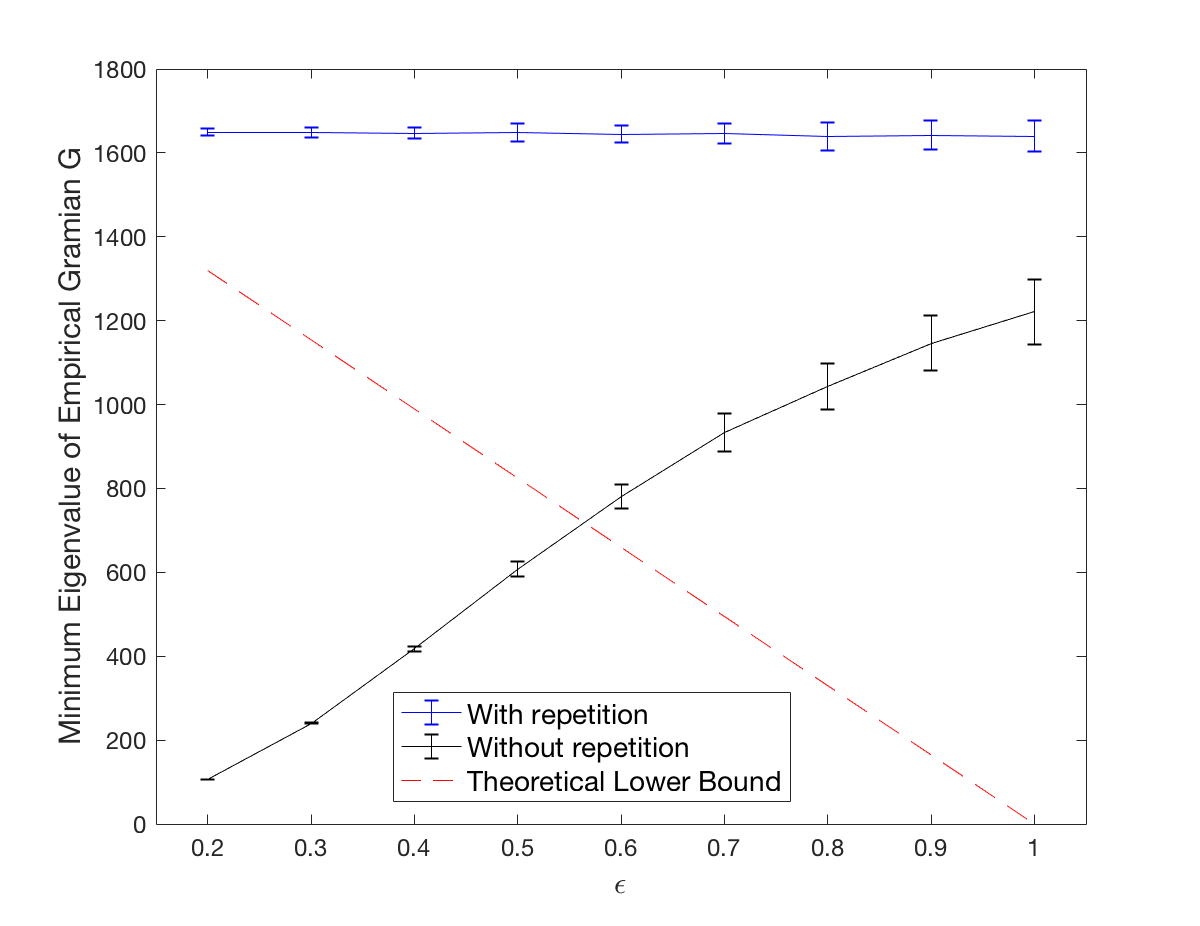}       
	\caption{Empirical comparison of Algorithm~\ref{algo:randsensor} with the minimum eigenvalue of the observability Gramian as metric with the bounds from Theorem~\ref{thm:randsensor}. Two variants of the algorithm are compared, one in which repetitions of sensors are allowed and another in which only unique sensors are considered. The error bars represent $\pm 1$ standard deviation.}
	\label{fig:mineig}
\end{figure}
\begin{figure}[h]
\centering
    \includegraphics[width = 0.6\columnwidth]{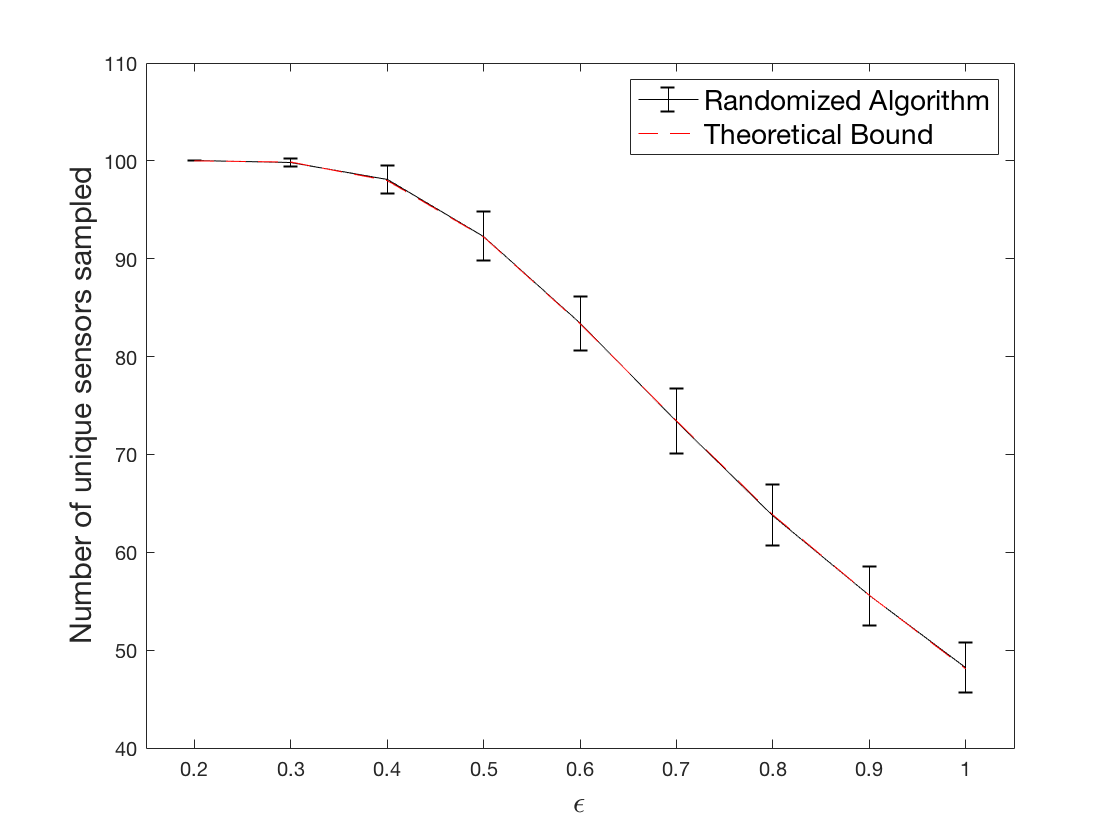}
    \caption{Number of unique sensors out of Algorithm~\ref{algo:randsensor} and the theoretical estimate from Proposition~\ref{thm:unique}.}
     \label{fig:unique}
     \end{figure}
\begin{figure}[h]
     \centering
    \includegraphics[width = 0.6\columnwidth]{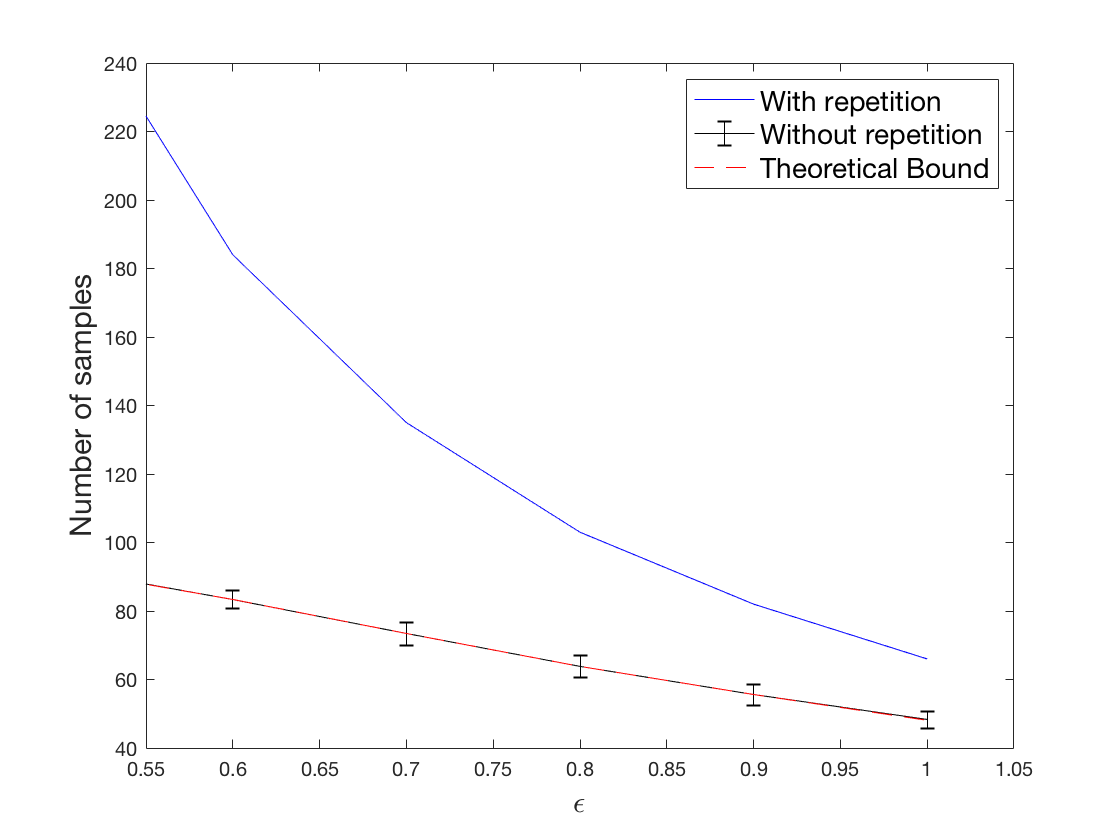}
     \caption{Comparison of the number of samples chosen by the two variants of Algorithm~\ref{algo:randsensor}.} 
         \label{fig:samples}
\end{figure}

\subsection{Comparison with a greedy heuristic}

We compared the performance of the two variants of Algorithm~\ref{algo:randsensor} to the following \emph{greedy} heuristic.
\begin{enumerate}
\item Sort the sensors in the decreasing order of $\mathcal{M}(W_i)$;
\item Keep selecting sensors until the metric evaluated over the selection exceeds $(1-\epsilon)\mathcal{M}(W)$;
\item Compare against the two variants of Algorithm~\ref{algo:randsensor} with the number of samples $c$ equal to the number of sensors $\subscr{c}{greedy}$ required by the greedy heuristic.
\end{enumerate}

\begin{figure}[h]
    \centering
       \includegraphics[width = 0.6\columnwidth]{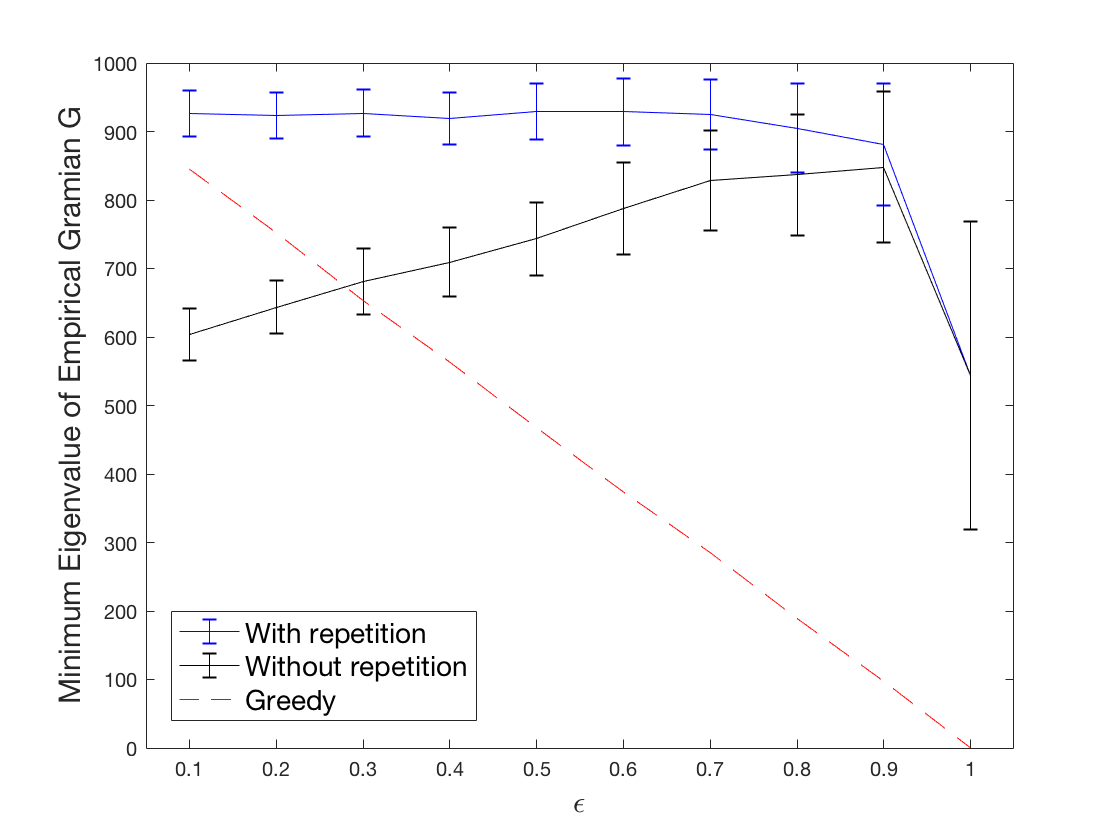}   
        \caption{Empirical comparison of the two variants of Algorithm~\ref{algo:randsensor} with a greedy heuristic in the minimum eigenvalue metric.}
         \label{fig:mineiggreedy}
    \end{figure}  
      \begin{figure}[h]
      \centering
       \includegraphics[width = 0.6\columnwidth]{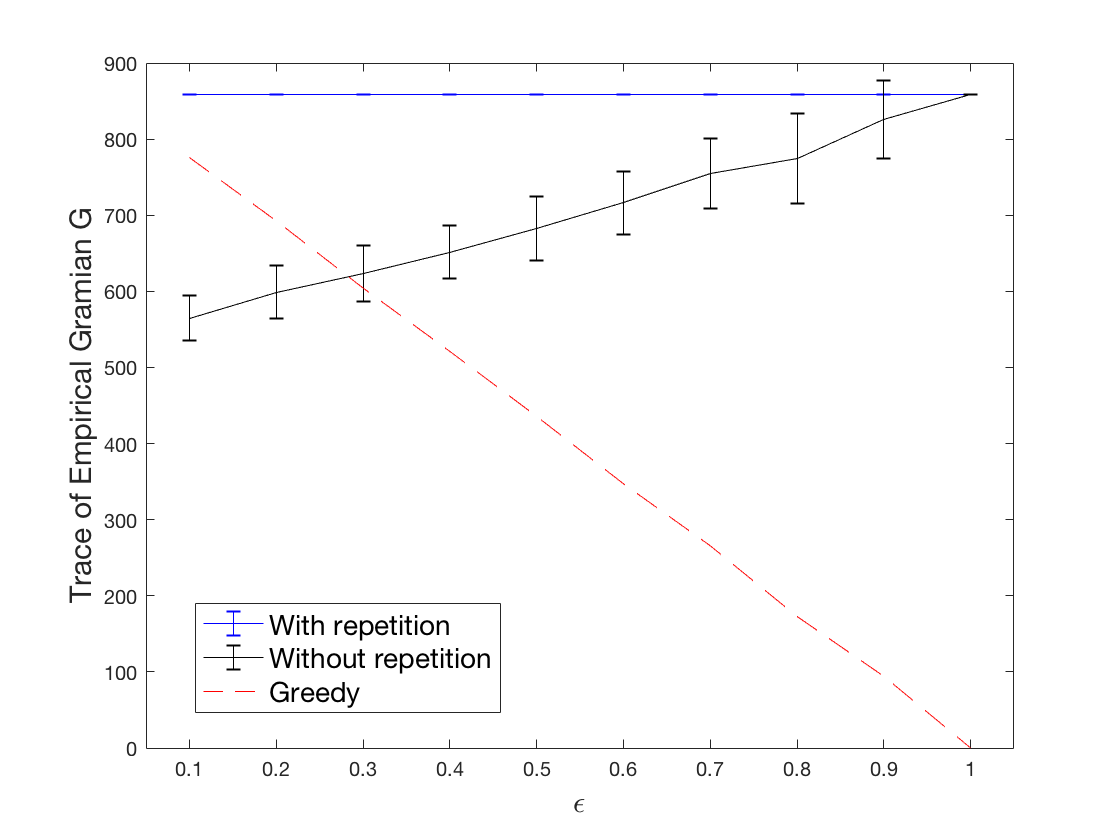}
        \caption{Empirical comparison of the two variants of Algorithm~\ref{algo:randsensor} with a greedy heuristic using the trace metric.}
       \label{fig:trace}
       \end{figure}
       \begin{figure}[h]
       \centering
       \includegraphics[width = 0.6\columnwidth]{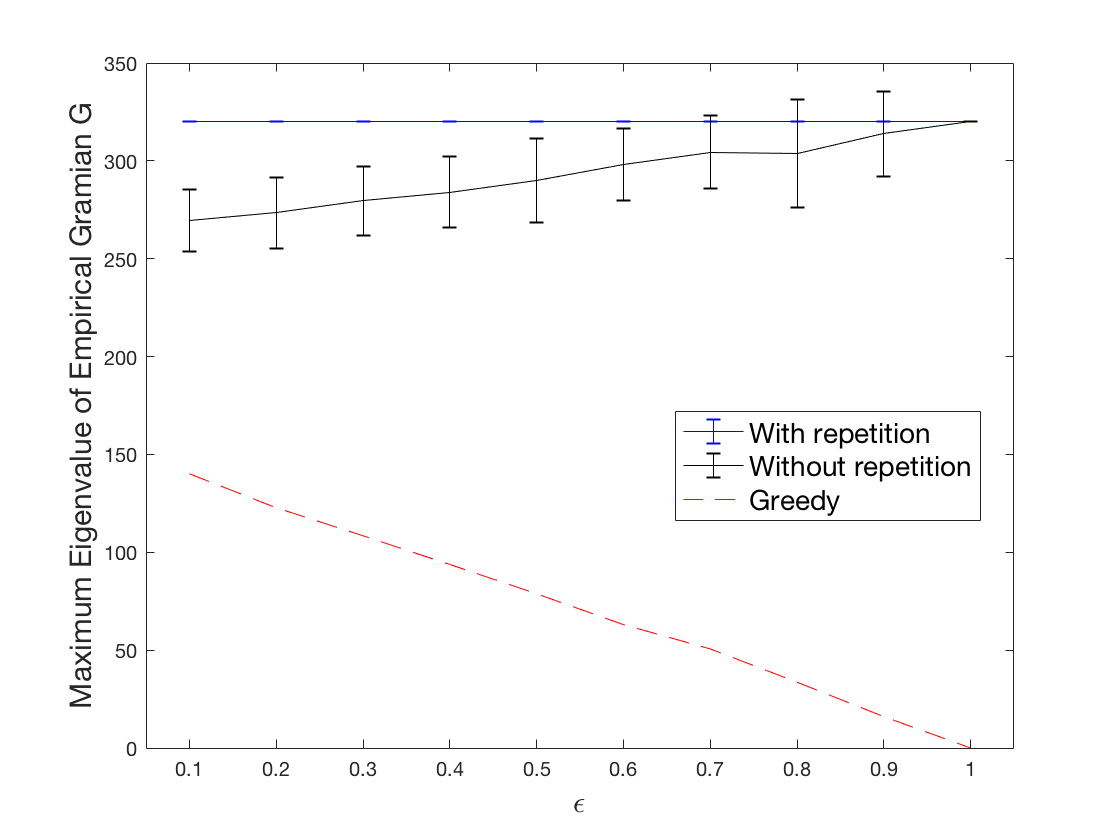}
    \caption{Empirical comparison of the two variants of Algorithm~\ref{algo:randsensor} with a greedy heuristic using the maximum eigenvalue metric.}
    \label{fig:maxeig}
\end{figure}

As reported in Figures~\ref{fig:mineiggreedy}, \ref{fig:trace} and \ref{fig:maxeig}, we observe that both variants of Algorithm~\ref{algo:randsensor} outperform the greedy algorithm for a wide range of $\epsilon$ for all three metrics. The variant in which repetitions are allowed outperforms the greedy heuristic over \emph{all} values of $\epsilon$. For the trace and the maximum eigenvalue, we normalized the $A$ matrix for numerical stability.

\begin{remark}[Alternate $C$ matrices]
We repeated the experiments with the choice of the $C$ matrix whose entries were drawn independently from the standard normal distribution. The results are observed to be qualitatively similar to the ones presented but the sample complexity bounds are observed to be about a factor of 2 higher. 
\end{remark}

\subsection{Recursive state estimation and Kalman filtering}
We now report the results of applying Algorithm~\ref{algo:randsensor} to least squares estimation as discussed in Section~\ref{sec:estimation}. We observe that the initial error covariance is an order of magnitude smaller than the theoretical upper bound derived in Proposition~\ref{prop:estimation}. Figure~\ref{fig:estimation} summarizes the result. 

\begin{figure}[t]
    \centering 
       \includegraphics[width = 0.6\columnwidth]{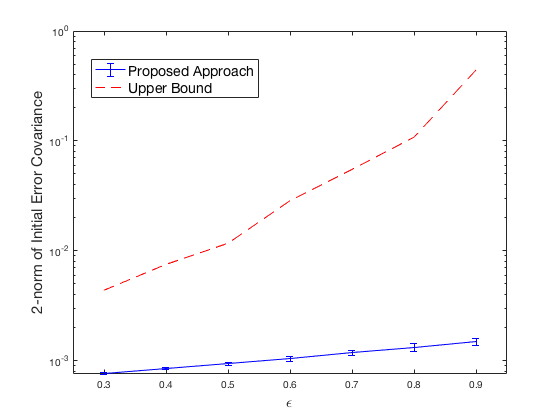}   
       \caption{Results of application of Algorithm~\ref{algo:randsensor} to recursive state estimation.}
       \label{fig:estimation}
 \end{figure}

We then applied the proposed approach to study the evolution of the steady state error covariance matrix $\bar{P}$ as discussed in Remark~\ref{rem:kf} and compare it with $P$. For varying values of $\epsilon$, Figure~\ref{fig:kf} reports the relative error 
\[
\frac{\norm{\bar{P} - P}}{\norm{P}}.
\]
We observe a fairly consistent nearly linear trend in the relative error with $\epsilon$. Future work will center on developing theoretical bounds to characterize this trend.

\begin{figure}[t]
    \centering 
       \includegraphics[width = 0.6\columnwidth]{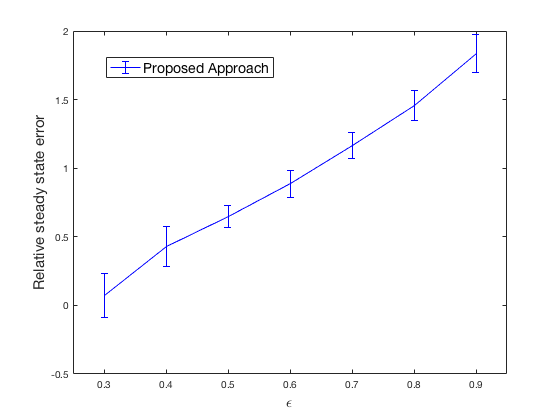}   
       \caption{Results of application of Algorithm~\ref{algo:randsensor} to Kalman filtering.}
       \label{fig:kf}
 \end{figure}

\section{Conclusion and Future Directions} \label{sec:conclusion}
This paper considered the problem of constructing an approximate linear dynamical system using only a subset of the sensors out of the total set such that the observability Gramian of the new system is approximately equal to that of the original system. We presented a novel randomized sampling algorithm that samples the sensors with replacement as per specified distributions. For specific metrics of the observability Gramian such as the trace or the maximum eigenvalue, we derived novel bounds on the number of samples required to yield a high probability lower bound on the metric evaluated on the approximate Gramian. With a larger number of samples, we then designed a different distribution to sample the sensors that leads to a high probability bound on the approximate Gramian and therefore, on any other standard metrics used in sensor selection such as the minimum eigenvalue or the trace of the Gramian inverse. This approach guarantees non-singularity of the approximate Gramian with high probability if the original system is observable. We demonstrated how the randomized procedure can be used for least squares estimation using fewer sensors than the original system and provide a high probability upper bound on the estimation error covariance. We supplemented the theoretical results with insightful numerical studies and comparisons with competing greedy approaches. Our approach outperforms the greedy algorithm in several scenarios.

This work suggests a number of open problems for future research. Examples include extensions to non-linear observability models using concepts from non-linear observability analysis, an improved analysis especially in the regime of low number of sensors and characterization of the error of the randomized algorithm relative to the best set of a given cardinality. A key feature of our proposed approach is that it is one-shot. Present efforts are on reducing the sample complexity requirements by developing sequential methods instead. Recent efforts in this direction includes \cite{jadbabaie2018} which addresses the problem of actuator scheduling. Another promising direction is a theoretical analysis of the application of the sampling procedure to bound the steady state error covariance matrix.


\appendix \label{sec:proof}
This section is dedicated to complete mathematical proofs of the theoretical statements from Section~\ref{sec:algorithm}. For increased clarity in the mathematical steps, we introduce the following matrices, 
\[
X_k := \begin{bmatrix}  c_k' & A'c_k' & \ldots & (A')^{T-1} c_k' \end{bmatrix} \in \real^{n \times T},
\]
for every $k \in \{1,\dots, m\}$ and $X:= \begin{bmatrix} X_1 & \dots & X_m \end{bmatrix}$. In short, we have $X_kX_k' = W_k$ and $XX' = W$.
 
\subsection{Proof of Theorem~\ref{thm:randsensor2}}
In order to derive a sharp sample complexity bound, we will require the following result on random matrices.

\begin{theorem}[Matrix Bernstein Inequality, \cite{tropp:11b}] \label{thm:bernstein}
Let $\X_1, \dots, \X_c$ be independent real symmetric random $n\times n$ matrices. Assume that, $\E[\X_j] = 0, \forall 1\leq j \leq c$, and that $\max_{1\leq j \leq c}\norm{\X_j}_2 \leq \rho_1$. Let $\norm{\sum_{j=1}^c\E[\X_j^2]}_2 \leq \rho_2$. Then, for any $\epsilon > 0$, 
\[
\Prob(\norm{\sum_{j=1}^c \X_j}_2 \geq \epsilon) \leq n \exp \Big(-\frac{\epsilon^2/2}{\rho_2 + \rho_1\epsilon/3} \Big),
\]
where the probability is defined with respect to the joint distribution of the $\X_j$'s.
\end{theorem}

\begin{theorem}[Theorem~2.1 from \cite{zhan:00}]\label{thm:diffsing}
If $B$ and $C$ are $n\times n$ real symmetric positive semi-definite matrices. Then, $\norm{B-C}_2 \leq \max\{\norm{B}_2, \norm{C}_2\}$.
\end{theorem}

\begin{proof}[Proof of Theorem~\ref{thm:randsensor2}]
The proof is inspired by the steps followed in the proof of Theorem~7.5 from \cite{holodnak:15} with the key difference that the Gramian $W$ is not a sum of rank 1 matrices in our case. We first change variables and define 
\[
\X_j := \frac{1}{cp_j}X_jX_j' - \frac{1}{c}XX'.
\]
\sdb{The distribution of the random matrix $\frac{1}{cp_j}X_jX_j'$ and therefore, of $\X_j$ is specified by the sampling probability $p_j$.} Observe that, by construction,
\[
\E[\X_j] = \frac{1}{c}(\sum_{k=1}^m X_kX_k' -  XX') = 0.
\]
Further,
\[
G - W = \sum_{j=1}^c \Big ( \frac{1}{cp_j}X_jX_j' - \frac{1}{c}XX' \Big ) = \sum_{j=1}^c \X_j.
\]
The first step is to show that $\norm{G - XX'}_2 \leq \bar{\epsilon}$ by showing that $\norm{\sum_{j=1}^c \X_j}_2 \leq \bar{\epsilon}$.

Since $\X_j$ is a difference of positive semi-definite matrices, applying Theorem~\ref{thm:diffsing}, we have
\begin{multline*}
\norm{\X_j}_2 \leq \max \{\frac{1}{cp_j}\norm{X_jX_j'}_2, \frac{1}{c}\norm{XX'}_2 \} \leq \hat{\rho_1}/c, \quad \hat{\rho}_1 := \max \Big \{ \max_{1\leq i \leq m} \Big \{ \frac{\norm{X_i}_2^2}{p_i} \Big \}, \norm{X}_2^2 \Big \}.
\end{multline*}

Next, we compute a bound for $\norm{\sum_{j=1}^c\E[\X_j^2]}_2$. We begin with
\[
\X_j^2 = (X_jX_j')^2 - \frac{1}{c}XX'(X_jX_j') - \frac{1}{c} (X_jX_j')XX' + \frac{1}{c^2}X(X'X)X'.
\]
Taking expectation, and using the fact that
\[
\E[X_jX_j'] = \frac{1}{c}XX',
\]
we obtain
\[
\E[\X_j^2] = \E[(X_jX_j')^2] - \frac{1}{c^2}X(X'X)X'.
\]
Thus, 
\begin{align*}
\sum_{j=1}^c \E[\X_j^2] &= \sum_{j=1}^c \E[(X_jX_j')^2] - \frac{1}{c^2}X(X'X)X'\\ &= \frac{1}{c} \Big ( \sum_{i=1}^m \frac{(X_iX_i')^2}{p_i} - X(X'X)X' \Big).
\end{align*}
Setting $X_i = XE_i$, where $E_i \in \real^{n\times mT}$ in which the $i$-th block rows contains the identity matrix of size $T\times T$, i.e., $E_i := \begin{bmatrix} 0_{T\times T} & \ldots & I_{T\times T} &\ldots & 0_{T\times T}  \end{bmatrix}$.

Thus, we have
\begin{align*}
\sum_{j=1}^c \E[\X_j^2] &= \frac{1}{c} X \Big (\sum_{i = 1}^m E_i \frac{(X_i'X_i)}{p_i} E_i' - X'X \Big) X'\\ &= \frac{1}{c}X(L - X'X)X',
\end{align*}
where $L := \blkdiag(E_1(X_1'X_1)E_1'/p_1, \dots, \\ E_m(X_m'X_m)E_m'/p_m)$. Taking norms and applying Theorem~\ref{thm:diffsing} to $\norm{L - X'X}_2$, we obtain
\begin{align}\label{eq:rho1}
&\norm{\sum_{j=1}^c \E[\X_j^2]}_2 \leq \frac{\norm{X}_2^2}{c} \max\{\norm{L}_2, \norm{X'X}_2 \}  \nonumber \\ &= \frac{\norm{X}_2^2}{c} \max\{\norm{X_1'X_1}_2/p_1, \dots, \norm{X_m'X_m}_2/p_m, \norm{X}_2^2 \} \nonumber \\ &= \frac{\norm{X}_2^2}{c} \hat{\rho}_1.
\end{align}

Setting, $p_i = \norm{X_iX_i'}_2 / \sum_{k=1}^m\norm{X_kX_k'}_2$, we obtain
\begin{align}\label{eq:rho2}
\hat{\rho}_1 &= \max\{\sum_{k=1}^m \norm{X_k X_k'}_2, \norm{X}_2^2 \} \leq \sum_{k=1}^m \norm{X_k X_k'}_2,
\end{align}
where we used the triangle inequality to conclude the final step. Combining \eqref{eq:rho1} and \eqref{eq:rho2} and applying Theorem~\ref{thm:bernstein}, we conclude that
\[
\Prob(\norm{\sum_{j=1}^c \X_j}_2 \geq \bar{\epsilon}) \leq n \exp \Big( \frac{-c\bar{\epsilon}^2}{2\hat{\rho}_1(\norm{X}_2^2 + \bar{\epsilon}/3)} \Big ).
\]
Setting the right hand side equal to $\delta$, and solving for $\bar{\epsilon}$, we obtain
\[
\bar{\epsilon} = \frac{\hat{\rho}_1}{3c}\log\frac{n}{\delta}  + \sqrt{\frac{\hat{\rho}_1}{3c}\log\frac{n}{\delta} (6\norm{X}_2^2  + \frac{\hat{\rho}_1}{3c}\log\frac{n}{\delta} )}.
\]
From \eqref{eq:rho2}, $\hat{\rho}_1 \leq \sum_{k=1}^m \norm{X_k X_k'}_2$, which yields
\begin{multline*}
\frac{\bar{\epsilon}}{\norm{X}_2^2} \leq  \frac{ \sum_{k=1}^m \norm{X_k X_k'}_2}{3c\norm{X}_2^2}\log\frac{n}{\delta} + \sqrt{ \frac{ \sum_{k=1}^m \norm{X_k X_k'}_2}{3c\norm{X}_2^2}\log\frac{n}{\delta} (6 +  \frac{ \sum_{k=1}^m \norm{X_k X_k'}_2}{3c\norm{X}_2^2}\log\frac{n}{\delta}) }.
\end{multline*}
The result is obtained by verifying the fact that for the particular choice of $c \geq 2.7\frac{\sum_{k=1}^m \norm{X_k X_k'}_2 \log(n/\delta)}{\epsilon^2\norm{X}_2^2}$, the above right hand side does not exceed $\epsilon$.
\end{proof}

\subsection{Proof of Theorem~\ref{thm:randsensor2}}
For $i \in \{1,\dots, n\}$ and $t \in \{1,\dots, c\}$, let \sdb{$Y_t^{ii}$} denote the $(i,i)$-th entry of the matrix $X_tX_t'/(cp_t)$. Let $X_t^i$ denote the $i$-th row of the matrix $X_t$. Then,
\[
\sdb{Y_t^{ii}} = \frac{1}{c}\frac{X_t^i {X_t^i}'}{p_t}.
\]
Therefore, 
\begin{align*}
\E[\sdb{Y_t^{ii}}] &= \frac{1}{c}\sum_{k=1}^m p_k \frac{X_k^i {X_k^i}'}{p_k} = \frac{1}{c} \sum_{k=1}^m [X_kX_k']_{ii} = \frac{1}{c}W_{ii}, \\
\end{align*}
Since the $(i,i)$-th entry of $G$ is given by $G_{ii} = \sum_{t=1}^c \sdb{Y_t^{ii}}$, we have $\E[G_{ii}] = \sum_{t=1}^c\E[\sdb{Y_t^{ii}}] = W_{ii}$. Therefore, it follows that $\E[\Tr(W-G)] = 0$. Next, we will show that the second moment of the random variable $\Tr(W-G)$ is also zero. We begin with
\begin{align*}
\E[\Tr^2(W - G)] &= \E[(\sum_{i=1}^n (W_{ii} - G_{ii}))^2] \\
&= \E[\sum_{i, j=1}^n (W_{ii} - G_{ii})(W_{jj} - G_{jj})] \\
&= \sum_{i, j=1}^n \E[(W_{ii} - G_{ii})(W_{jj} - G_{jj})] \\
&= -\sum_{i,j = 1}^n W_{ii}W_{jj} + \sum_{i,j = 1}^n \E[G_{ii}G_{jj}] \\
&= -\Tr^2(W) + \sum_{i,j = 1}^n \E[G_{ii}G_{jj}].
\end{align*}
Since $G_{ii} = \sum_{t=1}^c Y^{ii}_t$, $G_{jj} = \sum_{t=1}^c Y^{jj}_t$ and the random variables $Y^{ii}_t$ and $Y^{jj}_\tau$ are \emph{independent} for $t \neq \tau$ (see step 6 of Algorithm~\ref{algo:randsensor}), we have
\begin{align*}
\sum_{i,j = 1}^n \E[ G_{ii} G_{jj}] &= \sum_{i, j= 1}^n \E[ \sum_{t=1}^c Y^{ii}_t \sum_{\tau=1}^c Y^{jj}_\tau] \\
&= \sum_{i, j= 1}^n \Big ( \E [\sum_{t=1}^c Y^{ii}_t Y^{jj}_t] + \sum_{t, \tau = 1, t\neq \tau} \E[Y^{ii}_t Y^{jj}_\tau]  \Big)\\
&= \sum_{i, j= 1}^n \Big ( \sum_{t=1}^c \E[Y^{ii}_t Y^{jj}_t] + \sum_{t, \tau = 1, t\neq \tau} \E[Y^{ii}_t ]\E[Y^{jj}_\tau] \Big) \\
&= \sum_{i, j= 1}^n \Big ( \sum_{t=1}^c \E[Y^{ii}_t Y^{jj}_t] + \frac{1}{c^2} \sum_{t, \tau = 1, t\neq \tau}^c W_{ii}W_{jj} \Big) \\
&= \sum_{i, j= 1}^n \sum_{t=1}^c \E[Y^{ii}_t Y^{jj}_t] + \frac{(c^2-c)}{c^2} \sum_{i, j= 1}^n W_{ii}W_{jj} \\
&= \sum_{i, j= 1}^n \sum_{t=1}^c \E[Y^{ii}_t Y^{jj}_t] + \Big(1-\frac{1}{c}\Big) \Tr^2(W).
\end{align*}
Using the fact that
\begin{align*}
\E[Y^{ii}_t Y^{jj}_t] = \frac{1}{c^2} \sum_{k=1}^m \frac{(X^i_k{X^i_k}')(X^j_k{X^j_k}')}{p_k},
\end{align*}
and switching the summation over $t$ and $k$, we have
\begin{align*}
\sum_{i,j = 1}^n \E[ G_{ii} G_{jj}] &= \frac{1}{c^2} \sum_{i, j= 1}^n   \sum_{k=1}^m \sum_{t=1}^c \frac{(X^i_k{X^i_k}')(X^j_k{X^j_k}')}{p_k} + \Big(1-\frac{1}{c}\Big) \Tr^2(W) \\
&= \frac{1}{c} \sum_{k=1}^m \frac{1}{p_k} \sum_{i, j= 1}^n  (X^i_k{X^i_k}')(X^j_k{X^j_k}') + \Big(1-\frac{1}{c}\Big) \Tr^2(W) \\
&= \frac{1}{c} \sum_{k=1}^m \frac{1}{p_k} \Big (\sum_{i = 1}^n  (X^i_k{X^i_k}') \Big)^2 + \Big(1-\frac{1}{c}\Big) \Tr^2(W) \\
&= \frac{1}{c} \sum_{k=1}^m \frac{1}{p_k} \Tr^2(X_k{X_k}') + \Big(1-\frac{1}{c}\Big) \Tr^2(W). 
\end{align*}
Substituting $p_k := \Tr(X_k{X_k}')/ \sum_{k=1}^n\Tr(X_k{X_k}')$,
\begin{align*}
\E[\Tr^2(W - G)] &= \frac{1}{c} \sum_{k=1}^m \frac{1}{p_k} \Tr^2(X_k{X_k}') -\frac{1}{c} \Tr^2(W) \\
&= \frac{1}{c}( \sum_{k=1}^m \Tr (X_k{X_k}'))^2 -\frac{1}{c} \Tr^2(W) \\
&= \frac{1}{c} \Tr^2(W)-\frac{1}{c} \Tr^2(W)  = 0.
\end{align*}
The claim now follows upon application of \sdb{Cantelli's inequality\footnote{Cantelli's inequality \cite{ngo2011}: Given a real valued random variable $Z$ with mean $\mu$ and variance $\sigma^2$, $\Prob(Z - \mu \geq -\lambda) \geq 1 - \frac{\sigma^2}{\sigma^2 + \lambda}$, for any $\lambda > 0$.}.}

\subsection{Proofs of Theorem~\ref{thm:randsensor} and Proposition~\ref{thm:unique}}
The following matrix inequality (Corollary 2.2.2 from \cite{qiu2014}) will be useful in establishing Theorem~\ref{thm:randsensor}.

\begin{theorem}[Ahswede-Winter Inequality]\label{thm:chernoff}
Let $Z$ be an $n \times n$ random, symmetric, positive semi-definite matrix. Define $U=\E[Z]$ and suppose that $Z\leq \rho U$, for some scalar $\rho \geq 1$. Let $Z_1, \dots, Z_c$ denote independent copies of $Z$, i.e., independently sampled matrices with the same distribution as $Z$. For any $\epsilon \in (0,1)$, we have
\[
\Prob\Big ( (1-\epsilon) U \preceq \frac{1}{c}\sum_{k=1}^c Z_k \preceq (1+\epsilon) U \Big ) \geq 1 - 2n\e^{-\frac{\epsilon^2c}{4\rho}}.
\]
\end{theorem}

\begin{proof}[Proof of Theorem~\ref{thm:randsensor}]
Observe that, by construction, for every $i \in \{1,\dots,c\}$,
\[
\E[X_{j_i}X_{j_i}'] = \sum_{k=1}^m p_k \frac{1}{p_k} X_kX_k'  = XX' = W.
\]
So if we define the random variables $Z_{j_i} := X_{j_i}X_{j_i}'/p_{j_i}$, then
\begin{align*}
Z_{j_i} = \frac{\sum_{k=1}^m \gamma_k}{\gamma_{j_i}} W_{j_i} \leq \sum_{k=1}^m \gamma_k W =: \rho\E[Z_{j_i}],
\end{align*}
where we used the definition of $\gamma_{j_i}$ in the inequality. We now claim that the constant $\rho \geq 1$. Suppose that $\rho < 1$. Then from the definition of $\gamma_{k}$, we conclude that $W_k \leq \gamma_k W, \forall k \in \{1,\dots, m\}$. Summing over the $k$, we conclude that $W = \sum_{k=1}^m W_k \leq \rho W < W$, a contradiction.

Thus, we can verify that all assumptions of Theorem~\ref{thm:chernoff} are satisfied. Therefore, 
\[
\Prob\Big( (1-\epsilon)W \preceq G \preceq  (1+\epsilon)W \Big ) \geq 1- 2n\e^{-\frac{\epsilon^2c}{4\rho}}.
\]

The conclusion now follows by setting the right hand side equal to $1-\delta$, i.e.,
\[
2n\e^{-\frac{\epsilon^2c}{4\rho}} \leq  \delta \Leftrightarrow c \geq \frac{4\sum_{k=1}^m \gamma_k}{\epsilon^2}\log\frac{2n}{\delta}.
\]

 \end{proof}
 
 \begin{proof}[Proof of Proposition~\ref{thm:unique}]
 Let $U$ denote the number of unique sensors resulting from Algorithm~\ref{algo:randsensor}. Let $I_j \in \{0,1\}$ denote whether sensor $j$ does not get selected or gets selected, respectively. Then, for every $j\in \{1,\dots, m\}$,
\begin{multline*}
\E[I_j] = 0\times \Prob(I_j = 0) + 1 \times \Prob(I_j = 1) \\ = \Prob(I_j = 1) = 1 - \Prob(I_j = 0) = 1- \Big( 1 - p_j \Big)^c,
\end{multline*}
where the final step follows from the fact that the sensors are sampled independently and uniformly. Now, 
\begin{align*}
U = \sum_{j=1}^m I_j \Rightarrow \E[U] = \E \Big [ \sum_{j=1}^m I_j \Big] 
= \sum_{j=1}^m \E[I_j] = \sum_{j=1}^m 1- \Big( 1 - p_j \Big)^c 
&= m - \sum_{j=1}^m \Big( 1 - p_j \Big)^c\\
&= m - m \Big( 1 - \frac{1}{m} \Big)^c,
\end{align*}
in the case of uniform random sampling. Therefore, the expected ratio of the number of unique elements to the total number under uniform random sampling satisfies
\[
\E[U]= m- m\Big( 1 - \frac{1}{m} \Big)^c = m- m\Big(\Big( 1 - \frac{1}{m} \Big)^m\Big)^{\frac{c}{m}}\to m(1- e^{-\frac{c}{m}}),
 \]
  for large $m$. 

\end{proof}

\end{document}